\documentclass[pra,twocolumn,showkeys,preprintnumbers,superscriptaddress]{revtex4}

\usepackage{dcolumn}

\usepackage{amsmath,paralist,amsthm,comment,color,amssymb}

\input{Qcircuit}

\DeclareMathOperator{\tr}{tr}

\DeclareMathOperator{\supp}{supp}

\newcommand{\mc}[1]{\mathcal{#1}}

\newcommand{\C}{\mathbb{C}}
\newcommand{\F}{\mathbb{F}}
\newcommand{\nix}[1]{}

\newtheorem{theorem}{Theorem}
\newtheorem{corollary}[theorem]{Corollary}
\newtheorem{lemma}[theorem]{Lemma}

\newtheorem{remark}[theorem]{Remark}

\begin{document}

\title{Sharing Classical Secrets with CSS Codes}

\author{Pradeep Kiran Sarvepalli}
 \email{pradeep@phas.ubc.ca}
 \affiliation{Department of Physics and Astronomy\\
University of British Columbia, Vancouver V6T 1Z1, Canada
}
\author{Andreas Klappenecker}
 \email{klappi@cs.tamu.edu}
\affiliation{Department of Computer Science\\
Texas A\&M University, College Station TX 77843-3112, USA
}

\date{May 11, 2009}
\begin{abstract}
In this paper we investigate the use of quantum information to share
classical secrets.  While every quantum secret sharing scheme is a
quantum error correcting code, the converse is not true. Motivated by
this we sought to find quantum codes which can be converted to secret
sharing schemes. If we are interested in sharing classical secrets
using quantum information, then we show that a class of pure $[[n,1,d]]_q$
CSS codes can be converted to perfect secret sharing schemes. These
secret sharing schemes are perfect in the sense the unauthorized
parties do not learn anything about the secret. Gottesman had given
conditions to test whether a given subset is an authorized or
unauthorized set; they enable us to determine the access structure of
quantum secret sharing schemes. For the secret sharing schemes
proposed in this paper the access structure can be characterized in
terms of minimal codewords of the classical code underlying the CSS
code. This characterization of the access structure for quantum secret
sharing schemes is thought to be new. 
\end{abstract}

\maketitle

\section{Introduction}\label{sec:intro}
Quantum secret sharing schemes deal with the distribution of an
arbitrary secret state among $n$ parties (or shares) using quantum
states such that only authorized subsets can reconstruct the secret.
One can broadly classify quantum secret sharing schemes into the
category of schemes allowing one to (a) share quantum secrets and (b)
share classical secrets. This paper focuses on schemes of the second 
category (b). 

Quantum secret sharing schemes of both categories were introduced by
Hillery, Buzek, and Berthiaume~\cite{hillery99}. Classically, one can
always associate an error-correcting code to a perfect secret sharing
scheme---though determining the access structure of the associated
scheme is in general a very hard problem.  Additionally, one can also
derive a secret sharing scheme from a classical code, as was first
illustrated by the work of Massey \cite{massey93}.  It is not
surprising therefore that one finds connections between quantum secret
sharing schemes and quantum codes.  The connection between quantum
secret sharing schemes and quantum codes was first made explicit in
the work of Cleve et al. \cite{cleve99} and in greater depth by
Gottesman in \cite{gottesman00}. In fact  Cleve et al \cite{cleve99} showed that
quantum threshold schemes can be derived from quantum MDS codes and
gave an explicit method for these schemes. More recently, Rietjens et
al. \cite{rietjens05} showed that a $((k,2k-1))$ quantum threshold
scheme exists if and only if a $[[2k-1,1,k]]_q$ quantum MDS code
exists, thereby extending an earlier result found in \cite{cleve99}

However, the correspondence between quantum codes and quantum secret
sharing schemes does not appear to be as strong as in the classical
case. In general, one cannot derive a perfect quantum secret sharing
scheme from a quantum code.  However, if we restrict our secrets to
being classical, then we can relate a pure $[[n,1,d]]$ CSS code to a
secret sharing scheme, as we will show below.

Using quantum information to secure classical secrets has been studied
earlier in \cite{hillery99,karlsson99,gottesman00}. Gottesman had
given a convenient characterization of the access structure of secret
sharing schemes in \cite{gottesman00}. For the schemes proposed in
this paper we provide an alternative characterization which can
perhaps be extended to more quantum state sharing schemes.  The access
structures of the secret sharing schemes proposed in this paper can be
characterized using the notion of minimal codewords, a concept which
was introduced by Massey \cite{massey93} in the context of classical
secret sharing. Additionally, we draw upon the work of Gottesman
\cite{gottesman00} to link these ideas.

\subsection{Background: Quantum Secret Sharing} 
Let the parties of a secret sharing scheme be $P=\{ P_1,\ldots,P_n \}$.
Any subset of $P$ that can reconstruct the secret is called an authorized set. 
Subsets which cannot reconstruct the secret are called unauthorized sets. The collection
of authorized sets is called the access structure of the scheme and denoted by 
$\Gamma$. The collection of unauthorized sets is called the adversary structure and
denoted by $\mc{A}$. Clearly, $\Gamma \cup \mc{A} = 2^P$, the power set of $P$.
A minimal authorized set is one which can be used to reconstruct the secret but no proper
subset of which can reconstruct the secret. Clearly any subset which contains a minimal authorized set
is also authorized. The minimal access structure of the secret sharing scheme is the multiset consisting of 
minimal authorized sets. We denote a secret sharing scheme with (minimal) access structure $\Gamma_m$ as
$(\Sigma,\Gamma_m)$.

Of course, the secret sharing scheme must specify  more than the access structure. It must specify a 
means to encode the secret into the $n$ different shares and how any authorized set can recover the secret. 
In the language of quantum error correction these two tasks translate into encoding and decoding of a quantum
state which has been transmitted through a noisy quantum channel, in this case the quantum erasure channel. 
A secret sharing scheme is said to be perfect if 
\begin{compactenum}[i)]
\item an authorized set exactly reconstructs the secret 
\item an unauthorized cannot extract any information about the secret
\end{compactenum}

Essentially, any perfect secret sharing scheme must satisfy two requirements. On one hand, there is the requirement
of secrecy; any unauthorized set must know nothing about the secret. On the other hand, there is the requirement
of recoverability; any authorized set must be able to reconstruct the secret. One can also give a 
quantum information theoretic characterization of these requirements as was done in \cite{imai04}.
A characterization of these requirements for quantum secret sharing schemes can be found in \cite{gottesman00}, see also
\cite{cleve99}. 
In particular, for classical secrets this formulation is given as follows, see \cite[Theorem~9]{gottesman00}
for details.

\begin{lemma}[Gottesman]\label{lm:accessConditions}
Suppose we have a set of orthonormal states $\ket{\psi_i}$ encoding a classical secret. Then a set $T$
is an unauthorized set iff
\begin{eqnarray}
\bra {\psi_i} F \ket{\psi_i} =c(F)\label{eq:unauthSet}
\end{eqnarray}
independent of i for all operators $F$ on $T$. The set $T$ is authorized iff
\begin{eqnarray}
\bra {\psi_i} E \ket{\psi_j} =0 \quad (i\neq j)\label{eq:authSet}
\end{eqnarray}
for all operators $E$ on the complement of $T$.
\end{lemma}
We can state these conditions more informally. For an unauthorized set $T$, there is no measurement that can be
be performed on the qubits in the support of $T$ that can extract any information about the states $\ket{\psi_i}$.
Since an authorized set $T$ is to recover the secret, it can in effect correct erasures on the complement of 
$T$. If the conditions hold for any orthonormal basis of the space spanned by $\ket{\psi_i}$, then these states can
also be used for sharing quantum states, see for instance \cite[Theorem~1]{gottesman00}, \cite[Theorem~7]{cleve99}. 

\begin{remark}\label{rm:conditionsCSS}
We need not consider all operators on $T$, we only need to consider a basis of the operators on $T$. For $q$-ary
quantum systems schemes this basis of error operators can be identified with vectors in $\F_{q}^{2n}$. 
\end{remark}
We assume some background in (nonbinary) quantum codes, the reader
can refer to \cite{ketkar06} for more details.  Let $q$ be the power of a prime $p$. Let 
$B=\{ \ket{x}\mid x\in \F_q \}$ denote an orthonormal basis for
$\C^q$. For $a,b \in \F_q$, we define operators $X(a)$ and $Z(b)$
by 
\begin{eqnarray}
X(a)\ket{x}=\ket{x+a} \quad Z(b)\ket{x} =\omega^{\tr_{q/p}(bx)}\ket{x},
\end{eqnarray}
where $x\in \F_q$, $\omega = e^{j2\pi/p}$, and $j=\sqrt{-1}$.
These operators form a basis for error operators over a single qudit. 
Over $n$ qudits, we define the error operator 
$$
 X(a)Z(b) = X(a_1)Z(b_1)\otimes \cdots \otimes X(a_n)Z(b_n)
$$ 
for $a=(a_1,\dots,a_n)\in \F_q^n$ and 
$b=(b_1,\dots,a_b)\in \F_q^n$. The error  operators $\mc{E}= \{ X(a)Z(b) \mid  a, b \in \F_q^n\}$  form a basis for error
operators over $n$ qudits. We shall often denote $X(a)Z(b)$ by its representative over $\F_q^{2n}$ as $(a|b)$.
We say that an error operator $X(a)Z(b)$ has a support over $T \subseteq \{1,\ldots,n \}$ if 
$(a_t,b_t)\neq (0,0)$ for all $ t\in T$, and $(a_t,b_t)=(0,0)$ otherwise.

\section{Sharing Classical Secrets }
In this section we shall show that a pure $[[n,1,d]]$ CSS code can be
converted into a secret sharing scheme. We shall also characterize the
access structure of the scheme by using the notion of minimal
codewords.  Throughout this section we assume that the $[[n,1,d]]$
code under consideration has been derived from a classical code $C
\supseteq C^\perp$ with the parameters $[n,k,d]_q$ whose parity check
matrix is given as $H = \left[\begin{array}{cc} I_{n-k} &P
\end{array}\right]$. The dual code $C^\perp$ is defined as
$C^\perp=\{x\in \F_q^n \mid x\cdot c=0\mbox{ for all } c\in C \}$.
The stabilizer (matrix) of the CSS code is given as
\begin{eqnarray}
S= \left[\begin{array}{c|c}H &0\\0 &H\end{array}\right].\label{eq:cssStab}
\end{eqnarray}
Recall that the errors detectable by the CSS code are in $\F_q^{2n}\setminus (C\oplus C)$ or $C^\perp\oplus C^\perp$, where $C\oplus C$ is the direct sum of $C$ with itself. The undetectable errors are in 
$(C\oplus C )\setminus C^\perp\oplus C^\perp)$.

To define the minimal access structure of the secret sharing scheme we need the 
notion of minimal codewords. Let  $x,y \in \F_q^n $, then $x$ is said to cover $y$ if the support of 
$x$ contains the support of $y$. Alternatively, $y_i$ is zero whenever $x_i=0$, where we assume that
$x=(x_1,\ldots,x_n)$ and $y=(y_1,\ldots,y_n)$. A codeword $x \in C$ is said to be a minimal codeword 
if  
\begin{enumerate}[i)]
\item its left most  component is 1 and
\item it does not cover any other codeword of $C$ except scalar multiples of $x$. 
\end{enumerate}
A codeword which only satisfies ii) is said to be a minimal support. Every minimal codeword is of course
a minimal support. Minimal codewords were first introduced by
Massey in the context of classical secret sharing schemes, enabling a one to one correspondence with
minimal authorized sets. Minimal supports play an important role in studying the local equivalence of 
stabilizer states.  We also need the following lemma.

\begin{lemma}\label{lm:qdistNonbin}
Let $Q$ be a pure $[[n,1,d]]_q$ CSS code derived from $C^\perp \subseteq C\subseteq\F_q^n$.
For any two vectors $x,y\in C\setminus C^\perp$ we have $x\cdot y \neq 0$. If $q=2$, we
have $x\cdot y =1$ and $d$ odd.
\end{lemma}
\begin{proof}
Given the parameters of the quantum code the codes $C$ and $C^\perp$ must have the parameters
$[n,k,d]_q$ and $[n,n-k=k-1,d]_q$ respectively. Since $C\neq C^\perp$, it follows that there is 
at least one vector $c$ in $C\setminus C^\perp$, that satisfies $c\cdot c\neq 0$. Because $\dim C -\dim C^\perp=1$ 
we infer that $c$ and $C^\perp$ generate $C$. Therefore for any two vectors $x,y$ in $C\setminus C^\perp$
we can write them as $x=\alpha c+s_x$ and $y=\beta c+s_y$ for some $s_x,s_y\in C^\perp$ and $\alpha,\beta\in \F_q^{\times}$. Hence, $x\cdot y = (\alpha c+s_x)\cdot (\beta c+s_y)= \alpha\beta (c\cdot c) \neq0$.
If $q=2$, then it follows that $x\cdot y=1$. In particular $x\cdot x=1$, which implies that the weight of $x$ 
must be odd. Since the minimum distance depends on the weight of elements in $C\setminus C^\perp$, we conclude 
that $d$ is odd.
\end{proof}

\subsection{Proposed Secret Sharing Scheme}
First we shall describe the scheme and then show that it is indeed a valid secret sharing
scheme. 
\begin{theorem}\label{th:cssQss}
Let $Q$ be a pure $[[n,1,d]]_q$ CSS code derived from a classical code $C^\perp\subseteq C\subseteq \F_q^n$. 
Let $\mc{E}$ be the encoding given by the CSS code 
\begin{eqnarray}
\mc{E}:\ket{i}\mapsto \sum_{x\in C^\perp}\ket{x+i g}\quad i \in \F_q, \label{eq:encQss}
\end{eqnarray}
where $g \in C\setminus C^\perp$ and $g\cdot g=\beta$. Distribute  the $n$ qudits as the $n$ shares
for a secret sharing scheme, $\Sigma$. The minimal access structure  $\Gamma_m$ is given by 
\begin{eqnarray}
\Gamma_m = \left\{ \supp(c)\, |\,\begin{array}{l} c \mbox{ is a minimal codeword in } C\setminus C^\perp
\end{array} \right\}
\end{eqnarray}
Let $c=\alpha g+s_c$  be a minimal codeword for some $s_c\in C^\perp$.
The reconstruction for the authorized set $\supp(c)$ derived from $c$
is to compute 
\begin{eqnarray}
(\alpha\beta)^{-1}\sum_{j\in \supp(c)} c_j S_j, \label{eq:recovery}
\end{eqnarray} 
where $S_j$ is the share of the $j$th party.
\end{theorem} 
\begin{proof}
The proof of this theorem is a little long, so we shall break it into parts. First we shall show that the
minimal codewords define authorized sets i.e., they can recover the secret. Next we shall show that the associated
authorized sets are minimal. Thirdly, we shall show
that $\Gamma$ is complete i.e., every minimal authorized set is in $\Gamma$.
\begin{enumerate}[1)]
\item Recoverability: 
Let $c$ be a codeword in  $C\setminus C^\perp$, not necessarily minimal. 
Then $c$ can be written as  $c=\alpha g+s_c$ for some $s_c\in C$ and $\alpha\in \F_q^\times$. 
Adjoining an ancilla and computing the dot product with $c =\alpha g+s_c$ we get 
\begin{eqnarray*}
\ket{0}\ket{i g+C^\perp}&\mapsto & \sum_{x\in C^\perp}\ket{c\cdot x + c\cdot i g}\ket{x+ig},\\
&=&\sum_{x\in C^\perp}\ket{c\cdot x + \alpha g \cdot i g + s_c \cdot i g}\ket{x+ig},
\end{eqnarray*}
Since $c,g\in C \setminus C^\perp$ and $x,s_c\in C^\perp$ we have $c\cdot x=s_c\cdot ig =0$. 
Let $g\cdot g =\beta$, then by Lemma~\ref{lm:qdistNonbin},  $\beta\neq 0$ and is invertible in $\F_q$.
It follows
\begin{eqnarray*}
\ket{0}\ket{ig+C^\perp} &\mapsto& \ket{\alpha \beta i} \sum_{x\in C^\perp}\ket{x +i g}.
\end{eqnarray*}
Since both $\alpha$ and $\beta$ are known the secret can be recovered from the
ancilla which is in the state $\ket{\alpha\beta i}$. This proves that these subsets can reconstruct
the secret and they indeed define authorized sets. So every code word in $C\setminus C^\perp$ can define
an authorized set but it need not be minimal. Consequently every minimal codeword in $C\setminus C^\perp$ 
also defines an authorized set. 

\item 
Minimality of the authorized sets: Now let $c$ be a minimal codeword. We shall show that in this
case that any proper subset of $\supp(c)$ cannot reconstruct the secret. Let $T$ be a proper
subset of $\supp(c)$. Let the error operator $E=(a|b)$ be supported in $T$ where $a,b\in \F_q^n$, then  
$(a|b)$ cannot be a codeword in $C\oplus C$. Suppose it were a codeword in $C\oplus C$, then both $a,b \in C$.
Since $E$ is nontrivial at least one of $a$ and $b$ is nonzero and covered  by $c$, but then $c$ would not be a
minimal codeword. Therefore any error on  $T$  must be in $\F_q^{2n}\setminus (C\oplus C)$. 
But this means  that any such operator is detectable by the quantum code $Q$. If it is detectable, then it must not
reveal any information about the encoded states. In particular, it implies that $T$ satisfies 
equation~(\ref{eq:unauthSet}). Therefore every proper subset of $\supp(c)$ is an unauthorized set. This shows that $\supp(c)$ is a minimal authorized set.

\item 
Completeness of $\Gamma_m$:
Next we show that all minimal authorized sets are in $\Gamma_m$. Assume that there exists a minimal
authorized set $T$ which is not in $\Gamma_m$. Then $T$ must satisfy 
equation~(\ref{eq:authSet}). Additionally, $T$ fails to satisfy equation~(\ref{eq:unauthSet})
while every proper subset of $T$ being an unauthorized set does satisfy equation~(\ref{eq:unauthSet}).
This forces the existence of an operator $E=(a|b)$, with  $\supp(E)=T$,
that violates equation~(\ref{eq:unauthSet}). Now $E$  cannot be in 
$\F_q^n\setminus (C\oplus C) $ or $C^\perp\oplus C^\perp $, as these operators are detectable and cannot violate equation~(\ref{eq:unauthSet}). Therefore, $E$ must be in $(C\oplus C)\setminus (C^\perp\oplus C^\perp)$. 
Further, $(a|b) \in C\oplus C$ implies that $a,b\in C$. Now both $a,b$ cannot be in $C^\perp\subset C$
as then $(a|b)$ would be entirely in $C^\perp\oplus C^\perp$ and it would be detectable and cannot define
an authorized set. So at least one of $a,b$ is in $C \setminus C^\perp$. Without loss of generality let us
assume that $a\in C\setminus C^\perp$. But we already saw in step~1), that
any codeword in $C\setminus C^\perp$ defines an authorized set. So $\supp((a|0))=\supp(a)$ is itself an authorized set.
Since $(a|b)$ is a minimal authorized set,  $\supp(a)=\supp(a|b)=T$. 

Suppose that
$a$ is not a minimal codeword. Then there is some vector in $ C$ that is covered by $a$ and is not a scalar multiple
of $a$. First we show that there exists no $d\in C$ such that $\supp(d) \subsetneq T$. 
If $\supp(d)$ was a proper subset of $\supp(a)$, then $d$ cannot be in $C\setminus C^\perp$ as it would then
define an authorized set that is a proper subset of the minimal authorized set $T$. If $d$ is in $C^\perp$, then
there exists a linear combination of $a$ and $d$ with support strictly a subset of $T$. Further this linear 
combination is also in $C\setminus C^\perp$ and by step~1) it would define an authorized set violating the minimality
of $T$. Therefore any $d\in C$  covered by $a$ and not a scalar multiple of $a$ must have $\supp(d)=T$. But 
this implies that $C$ contains a linear combination of $a$ and $d$ with support strictly less than $T$ violating our previous conclusion that there exists no such element in $C$. Therefore $a$ is a minimal codeword of $C$ and it lies in $C\setminus C^\perp$. (If the left most component of $a$ is not 1 we can choose a scalar multiple of it so that 
it is 1. In any case, $a$ and its scalar multiples have same support and they correspond to the same (minimal)
authorized set).
\end{enumerate}
\end{proof}

Since a codeword of minimum distance does not cover any other codeword,  there always exists a
scalar multiple of it which is a minimal codeword. Therefore, the minimal access structure always
contains the sets corresponding to the support of the every minimum distance
codeword in $C\setminus C^\perp$. 

\begin{corollary}
In the secret sharing scheme specified in Theorem~\ref{th:cssQss}, the support of every minimum distance 
codeword in $C\setminus C^\perp$ gives rise to a minimal authorized set.
\end{corollary}

If $q=2$, then we can simplify the reconstruction process, we only need to take the parity of the parties
in the minimal authorized set. 

\begin{corollary}
Let $Q$ be a pure $[[n,1,d]]_2$ CSS code derived from a classical code $C^\perp\subseteq C\subseteq \F_2^n$. 
Let $\mc{E}$ be the encoding given by the CSS code 
\begin{eqnarray}
\mc{E}:\ket{i}\mapsto \sum_{x\in C^\perp}\ket{x+i g}\quad i \in \F_2, \label{eq:encQssBin}
\end{eqnarray}
where $g \in C\setminus C^\perp$. Distribute  the $n$ qubits as the $n$ shares
for a secret sharing scheme, $\Sigma$. The minimal access structure  $\Gamma_m$ is given by 
\begin{eqnarray}
\Gamma_m = \left\{ \supp(c)\, |\,\begin{array}{l} c \mbox{ is a minimal codeword in } C\setminus C^\perp
\end{array} \right\}
\end{eqnarray}
The reconstruction for an authorized set is to simply compute the parity of the set (into an ancilla).
\end{corollary}

The secret can be encoded using the encoding methods of CSS codes, see \cite{grassl03}.
Reconstructing the secret for these schemes is extremely simple as shown below.  
We will need the multiplier gate $M(c)$ and the  generalized CNOT gate, $A$ shown below.

\[
\Qcircuit @C=1em @R=.7em {
& \gate{c}&\qw && \ctrl{1}&\qw  \\
& &  & &\targ &\qw & \\
& &  & & & & & &&&\\
& \text{i)}&  & &\text{ii)} & 
}
\]
Their action on the basis states is given as:
\begin{compactenum}[i)]
\item $M(c)\ket{x}=\ket{cx}, c\in \F_q^\times$
\item $A\ket{x} \ket{y} = \ket{x}\ket{x+y}$
\end{compactenum}

The recovery as given in equation~(\ref{eq:recovery}) is computed by performing the following
operation for each $c_j\neq 0$.
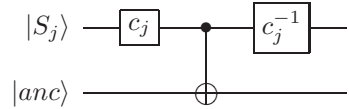
\[
\Qcircuit @C=1.4em @R=1.2em {
\lstick{\ket{S_j}}&\gate{c_j}&\ctrl{1} &\gate{c_j^{-1}}&\qw&\\
\lstick{\ket{anc}}&\qw&\targ&\qw&\qw&
}
\]
The final scaling by $(\alpha\beta)^{-1}$ can be done classically.
\subsection{Illustration}
We illustrate the strategy using a $[[11,1,3]]$ CSS code \cite{grassl00}
 it can be derived from a 
code $C$ with the following generator and parity check matrices. 
\begin{eqnarray}
G&=& \left[\begin{array}{ccccccccccc} 
1& 0& 0& 0& 0& 0& 0& 1& 0& 0& 1 \\
 0& 1& 0& 0& 0& 0& 1& 1& 1& 1& 1 \\
 0& 0& 1& 0& 0& 0& 0& 0& 0& 1& 1 \\
 0& 0& 0& 1& 0& 0& 1& 0& 0& 0& 1 \\
 0& 0& 0& 0& 1& 0& 1& 1& 1& 1& 0 \\
0& 0& 0& 0& 0& 1& 0& 0& 1& 0& 1 
\nix{  1& 0& 0& 0& 0& 1& 0& 1& 1& 0& 0 \\
 0& 1& 0& 0& 0& 0& 1& 1& 1& 1& 1 \\
 0& 0& 1& 0& 0& 1& 0& 0& 1& 1& 0 \\
 0& 0& 0& 1& 0& 1& 1& 0& 1& 0& 0 \\
 0& 0& 0& 0& 1& 1& 1& 1& 0& 1& 1 \\\hline
0& 0& 0& 0& 0& 1& 0& 0& 1& 0& 1 
}
\end{array}\right]\\
H&=&\left[ \begin{array}{ccccccccccc} 
 1& 0& 0& 0& 0& 1& 0& 1& 1& 0& 0 \\
 0& 1& 0& 0& 0& 0& 1& 1& 1& 1& 1 \\
 0& 0& 1& 0& 0& 1& 0& 0& 1& 1& 0 \\
 0& 0& 0& 1& 0& 1& 1& 0& 1& 0& 0 \\
 0& 0& 0& 0& 1& 1& 1& 1& 0& 1& 1 
\end{array}\right]
\end{eqnarray}
Let us encode the secret 
\begin{eqnarray}
\ket{s} &\mapsto& \sum_{c\in C^\perp} \ket{c+s e},
\end{eqnarray}
where 
$e=\left[\begin{array}{ccccccccccc}0&0 & 0&0 &0 &1& 0& 0& 1&0 & 1 \end{array}\right]$.
The secret sharing scheme assumes that we distribute each qubit as a  share.
The minimal access structure of the secret sharing scheme is given by $\Gamma_m$.

$$
\Gamma_m = \left\{ \begin{array}{c}
\{ 3, 10, 11 \};
\{ 6, 9, 11 \};
\{ 4, 7, 11 \};
\{ 2, 5, 11 \};\\
\{ 1, 8, 11 \};
\{ 2, 3, 4, 6, 8 \};\\
\{ 4, 5, 6, 8, 10 \};
\{ 1, 3, 4, 5, 6 \};
\{ 1, 2, 4, 6, 10 \};\\
\{ 3, 4, 5, 8, 9 \};
\{ 2, 4, 8, 9, 10 \};
\{ 1, 2, 3, 4, 9 \};\\
\{ 1, 4, 5, 9, 10 \};
\{ 3, 5, 6, 7, 8 \};
\{ 2, 6, 7, 8, 10 \};\\
\{ 1, 2, 3, 6, 7 \};
\{ 1, 5, 6, 7, 10 \};
\{ 5, 7, 8, 9, 10 \};\\
\{ 2, 3, 7, 8, 9 \};
\{ 1, 3, 5, 7, 9 \};
\{ 1, 2, 7, 9, 10 \}
\end{array}
\right\}
$$

It can be checked that the parity of any of these subsets will give $s$.
Further, any subset that contains an element of $\Gamma_m$ as a subset can also 
perform reconstruction. Please note that this is not a threshold scheme, there
exist minimal authorized sets of size three and five.

\section{Conclusion}
In this paper we have given new methods to share classical secrets
using quantum information. We have been able to strengthen the connection
between quantum secret sharing schemes and quantum error correcting
codes and given a new characterization of the access structure in
terms of minimal codewords. This characterization is potentially of
larger applicability, and its extension to additive quantum codes and
quantum secrets will be explored elsewhere.

\end{document}